\documentclass[oribibl]{llncs}

\usepackage[algo2e,titlenumbered]{algorithm2e}
\usepackage[utf8]{inputenc}
\usepackage{url}
\usepackage{amsfonts}
\usepackage{amsmath}
\usepackage{amssymb}
\usepackage{algorithm}
\usepackage{algorithmic}
\usepackage{varioref}
\usepackage{hyperref}
\usepackage{bookmark}
\usepackage{pdflscape}
\usepackage{graphicx,color}
\usepackage{caption}
\usepackage{graphics}
\usepackage{setspace}
\usepackage{ifsym}

\makeatletter
\newcommand*{\rom}[1]{\expandafter\@slowromancap\romannumeral #1@}
\makeatother

\begin{document}

\title{On the Finite Optimal Convergence of Logic-Based Benders' Decomposition in Solving 0-1 Min-max Regret Optimization Problems with Interval Costs}

\author{Lucas Assun\c{c}\~ao\inst{1}$^,$\thanks{Partially supported by the Coordination for the Improvement of Higher Education
Personnel, Brazil (CAPES).}$^,$\thanks{The author thanks Vitor A. A. Souza and Phillippe Samer for the valuable discussions throughout the
conception of this work.}
\and Andr\'ea Cynthia Santos \inst{2} \and Thiago F. Noronha \inst{3} \and Rafael Andrade \inst{4}}
\institute{Departamento de Engenharia de Produção, Universidade Federal de Minas Gerais,
Avenida Antônio Carlos, 6627, CEP 31.270-901, Belo Horizonte, MG, 
Brazil \email{lucas-assuncao@ufmg.br}
\and
ICD-LOSI, UMR CNRS 6281, Universit\'e de Technologie de Troyes,\\
12, rue Marie Curie, CS 42060, 10004, Troyes CEDEX, France
\email{andrea.duhamel@utt.fr}
\and 
Departamento de Ciência da Computação, Universidade Federal de Minas Gerais,
Avenida Antônio Carlos, 6627, CEP 31.270-901, Belo Horizonte, MG, 
Brazil \email{tfn@dcc.ufmg.br}
\and Departamento de Estatística e Matemática Aplicada, Universidade Federal do Cear\'a,
Campus do Pici, BL 910, CEP 60.455-900, Fortaleza, CE, Brazil \email{rca@lia.ufc.br}}

\maketitle

\begin{abstract}
This paper addresses a class of problems under interval data uncertainty composed of
\emph{min-max regret} versions of classical 0-1 optimization problems with interval costs. We refer to them as \emph{interval 0-1 min-max regret} problems.
The state-of-the-art exact algorithms for this class of problems work by solving a corresponding mixed integer linear programming formulation
in a Benders' decomposition fashion. Each of the possibly exponentially many Benders' cuts is separated on the fly through the resolution of an instance of
the classical 0-1 optimization problem counterpart. Since these separation subproblems may be NP-hard, not all of them can be modeled by means of linear
programming, unless P = NP. In these cases, the convergence of the aforementioned algorithms are not guaranteed in a straightforward manner.
In fact, to the best of our knowledge, their finite convergence has not been explicitly proved for any interval 0-1 min-max regret problem.
In this work, we formally describe these algorithms through the definition of a logic-based Benders' decomposition framework and prove their convergence
to an optimal solution in a finite number of iterations. As this framework is applicable to any interval 0-1 min-max regret problem, its finite optimal
convergence also holds in the cases where the separation subproblems are NP-hard.
\end{abstract}

\section{Introduction}
Robust Optimization (RO) \cite{Kouvelis1997} has drawn particular attention as an alternative to stochastic programming \cite{Spall03} in modeling uncertainty.
In RO, instead of considering a probabilistic description known \emph{a priori}, the variability of the data is represented by deterministic values in the context of \emph{scenarios}.
A \emph{scenario} corresponds to a parameters assignment, \emph{i.e.}, a value is fixed for each parameter subject to uncertainty.
Two main approaches are adopted to model RO problems: the \emph{discrete scenarios model} and the \emph{interval data model}.
In the former, a discrete set of possible scenarios is considered. In the latter, the uncertainty referred to a parameter is represented by a
continuous interval of possible values. Differently from the discrete scenarios model, the infinite many possible scenarios that arise
in the interval data model are not explicitly given.
Nevertheless, in both models, a classical (\emph{i.e.}, parameters known in advance) optimization problem takes place
whenever a scenario is established.

The most commonly adopted RO criteria are
the \emph{absolute robustness} criterion, the \textit{min-max regret} (also known as \emph{robust deviation} criterion) and
the \textit{min-max relative regret} (also known as \emph{relative robustness} criterion).
The absolute robustness criterion is based on the anticipation of the worst possible conditions.
Solutions for RO problems under such criterion tend to be conservative, as they optimize only a worst-case scenario.
On the other hand, the min-max regret and the min-max relative regret are less conservative criteria and, for this reason, they have been addressed in
several works (\emph{e.g.}, \cite{Averbakh05,amadeu14,Mo06,MontemanniBarta07,Pereira11}).
Intuitively speaking, the \emph{regret (robust deviation)} of a solution in a given scenario is the cost difference between such solution and
an optimal one for this scenario. 
In turn, the \emph{relative regret} of a solution in a given scenario consists of the corresponding regret normalized by the cost of an optimal solution
for the scenario considered.
The \emph{(relative) robustness} cost of a solution is defined as its maximum \emph{(relative) regret} over all scenarios.
In this sense, the min-max (relative) regret criterion aims at finding a solution that has the minimum
(relative) robustness cost. Such solution is referred to as a \emph{robust solution}.

RO versions of several combinatorial optimization problems have been studied in the literature, addressing, for example, uncertainties over costs.
Handling uncertain costs brings an extra level of difficulty, such that even polynomially solvable problems become NP-hard in their corresponding
robust versions \cite{Kasperski08,MonteGa04c,Pereira11,YaKaPi01}.
In this study, we are interested in a particular class of RO problems, namely \emph{interval 0-1 min-max regret} problems,
which consist of min-max regret versions of Binary Integer Linear Programming (BILP) problems with interval costs.
Notice that a large variety of classical optimization
problems can be modeled as BILP problems, including (i) polynomially solvable problems, such as the shortest path problem, the minimum spanning tree
problem and the assignment problem, and (ii) NP-hard combinatorial problems, such as the 0-1 knapsack problem, the set
covering problem, the traveling salesman problem and the restricted shortest path problem \cite{Garey79}.
An especially challenging subclass of interval 0-1 min-max regret problems, referred to as \emph{interval 0-1 robust-hard problems}, arises when
we address interval 0-1 min-max regret versions of classical NP-hard combinatorial problems as those aforementioned in (ii). 

Aissi at al.~\cite{Aissi09} showed that, for any interval 0-1 min-max regret problem (including interval 0-1 robust-hard problems),
the robustness cost of a solution can be computed by solving a single instance of the classical optimization problem counterpart (\emph{i.e.}, costs known in
advance) in a particular scenario. Therefore, one does not have to consider all the infinite many possible scenarios during the search for a robust solution,
but only a subset of them, one for each feasible solution. Nevertheless, since the number of these promising scenarios can still be huge,
the state-of-the-art exact algorithms for interval 0-1 min-max regret problems work by implicitly separating them on the fly, in a
Benders' decomposition \cite{Benders62} fashion (see, \emph{e.g.}, \cite{Mo06,MontemanniBarta07,Averbakh13}). Precisely, each Benders' cut is generated through the resolution of an instance of the classical optimization problem counterpart.
Notice that, for interval 0-1 robust-hard problems, these separation subproblems are NP-hard and, thus, they cannot be modeled by means of Linear
Programming (LP), unless P = NP.
In these cases, the convergence of the aforementioned algorithms are not guaranteed in a straightforward manner.

These exact algorithms, which have their roots in \emph{logic-based Benders' decomposition} \cite{Hooker03} (also see \cite{Codato06}),
have been successfully applied to solve several interval 0-1 min-max regret problems (\emph{e.g.}, \cite{Mo06,MonteGa04c,Pereira11}),
including interval 0-1 robust-hard problems, such as the robust traveling salesman problem \cite{MontemanniBarta07} and the robust set covering
problem \cite{Averbakh13}.
However, to the best of our knowledge, their convergence has not been explicitly proved for any interval 0-1 min-max regret problem.
In this work, we formally describe these algorithms through the definition of a logic-based Benders' decomposition framework and prove their finite convergence
to an optimal solution. Precisely, we show, by contradiction, that a new cut is always generated per iteration, in a finite space of possible solutions.
As the framework is applicable to any interval 0-1 min-max regret problem, its finite optimal
convergence also holds in solving interval 0-1 robust-hard problems, \emph{i.e.}, in the cases where the separation subproblems are NP-hard.

The remainder of this work is organized as follows. The Benders' decomposition method is briefly introduced in Sect.~\ref{s_logic_benders},
followed by the description of a standard modeling technique for interval 0-1 min-max regret problems (Sect.~\ref{s_model}).
In addition, a generalization of state-of-the-art exact algorithms for interval {0-1} min-max regret problems is devised through the description of a
logic-based Benders' decomposition framework (Sect.~\ref{s_benders}).
The finite convergence of these algorithms to optimal solutions is proved in the same section, and concluding remarks are given in the last section.

\section{A Logic-Based Benders' Decomposition Framework for Interval 0-1 Min-max Regret Problems}
\label{s_logic_benders}

The \emph{classical Benders' decomposition method} \cite{Benders62} was originally proposed to tackle Mixed Integer Linear Programming (MILP) problems of the form
$P: \min\{cx + dy:\, Ax + By \geq b,\, x \in \mathbb{Z}_+^{n_1},\, y \geq \textbf{0}\}$. In this case,
there are $n_1$ integer variables and $n_2$ continuous ones, which are represented by the column vectors $x$ and $y$, respectively,
and their corresponding cost values are given by the row vectors $c$ and $d$. Moreover, $b$ is an $m$-dimensional column vector, and $A$ and $B$ are $m \times n_1$ and $m \times n_2$
matrices, respectively.
Given a vector $\bar{x} \in \mathbb{Z}_+^{n_1}$, the \emph{classical Benders' reformulation} starts by defining an LP \emph{primal subproblem} $PS(\bar{x}): \min\{dy:\, By \geq b - A\bar{x},\, y \geq \textbf{0}\}$ through the projection of the continuous variables $y$
in the space defined by $\bar{x}$.
Notice that $PS(\bar{x})$ can be represented by means of the corresponding \emph{dual subproblem} $DS(\bar{x}): \max\{\mu(b-A\bar{x}):\, \mu B \leq d,\, \mu \geq \textbf{0}\}$,
where $\mu$ is an $m$-dimensional row vector referred to the dual variables.

Let $EP(\bar{x})$ and $ER(\bar{x})$ be, respectively, the sets of extreme points and extreme rays of a given $DS(\bar{x})$ subproblem.
One may observe that the feasible region of $DS(\bar{x})$ does not depend on the value assumed by $\bar{x}$.
Thus, hereafter, these sets are referred to as $EP$ and $ER$ for all $\bar{x} \in \mathbb{Z}_+^{n_1}$.

Considering a nonnegative continuous variable $\rho$, the resolution of each dual subproblem $DS(\bar{x})$ leads to a new linear
constraint (i) $\rho \geq {{\bar{\mu}}}(b-Ax)$, if $DS(\bar{x})$ has a bounded
optimal solution ${\bar{\mu}} \in EP$, or (ii) ${{\bar{\nu}}}(b-Ax) \leq 0$, if $DP$ has an unbounded solution, represented
by ${\bar{\nu}} \in ER$. The Benders' cuts described in (i) are called \emph{optimality cuts}, whereas the ones in (ii) are called
\emph{feasibility cuts}. Both types of cuts are used to populate on the fly a \emph{reformulated problem}, defined as:
\begin{eqnarray}
  \mbox{$(RP)\quad$}\min &&\{cx + \rho\}\\
  s.t.  && \rho \geq {{\bar{\mu}}}(b-Ax) \quad \forall\, {\bar \mu} \in EP, \label{rp01}\\
	&& {{\bar{\nu}}}(b-Ax) \leq 0 \quad \forall\, {\bar \nu} \in ER, \label{rp02}\\
	&& \rho \geq 0,\\
	&& x \in \mathbb{Z}_+^{n_1}.
\end{eqnarray}

Let $MRP$ be a relaxed $RP$ problem, called \emph{master problem}, which considers only a subset of the extreme points and extreme rays associated with
constraints (\ref{rp01}) and (\ref{rp02}), respectively.
At each iteration of the \emph{classical Benders' decomposition algorithm}, the master problem is solved,
obtaining a solution $(\bar{x}, \bar{\rho}) \in \mathbb{Z}_+^{n_1} \times \mathbb{R}_+$.
If $(\bar{x}, \bar{\rho})$ is not feasible for the original reformulated problem $RP$,
a corresponding $DS(\bar{x})$ subproblem is solved in order to generate a new Benders' cut,
either a feasibility cut or an optimality one. The new cut is added to the master problem $MRP$ and the algorithm iterates until a
feasible (and, therefore, optimal) solution for $RP$ is found.

The finite convergence of the classical Benders' decomposition method is guaranteed by the fact that the polyhedron referred to any LP problem can be described by finite sets of extreme points and extreme rays,
and that a new Benders' cut is generated per iteration of the algorithm. We refer to \cite{Benders62} for the detailed proof.
Methodologies able to improve the convergence of the method were studied in several works (see, \emph{e.g.}, \cite{Fischetti10,Magnanti81,McDaniel77}).
In addition, nonlinear convex duality theory was later applied to devise a generalized approach, namely \emph{generalized Benders' decomposition},
applicable to mixed integer nonlinear problems \cite{Geoffrion72}.

More recently, Hooker and Ottosson \cite{Hooker03} introduced the idea of the so-called \emph{logic-based Benders' decomposition},
a Benders-like decomposition approach that is suitable for a broader class of problems. In fact, the latter approach is intended to tackle any
optimization problem by exploring the concept of an \emph{inference dual subproblem}. In particular, that is the problem of inferring a strongest possible bound for
a set of constraints from the original problem that are relaxed in the master problem.
Notice that the aforementioned inference subproblems are not restricted to linear and nonlinear continuous problems.
In fact, they can even be NP-hard combinatorial problems (see, \emph{e.g.}, \cite{Cote14}).
Therefore, the convergence of the logic-based Benders' decomposition method cannot be showed in a straightforward manner
for all classes of optimization problems.
As pointed out in \cite{Hooker03}, the convergence of the method relies on some peculiarities
of the (logic-based) Benders' reformulation, such as the way the inference dual subproblems are devised and the finiteness of the search space referred
to them.

In the remainder of this section, we describe a framework that generalizes state-of-the-art logic-based Benders' decomposition algorithms widely used to solve interval 0-1
min-max regret problems \cite{Mo06,MontemanniBarta07,MonteGa04c,Pereira11,Averbakh13}. In addition, we show its convergence to an optimal solution
in a finite number of iterations. The framework addresses MILP formulations with typically
an exponential number of constraints, as detailed in the sequel.

\subsection{Mathematical Formulation}
\label{s_model}
Consider $\mathcal{G}$, a generic BILP minimization problem defined as follows.
\begin{eqnarray}
  \mbox{$(\mathcal{G})\quad$}\min && cx \label{bilp00}\\
  s.t.  && Ax \geq b, \label{bilp01}\\
	&& x \in \{0,1\}^n. \label{bilp02}
\end{eqnarray}

The binary variables are represented by an $n$-dimensional column vector $x$, whereas their corresponding cost values are given by an
$n$-dimensional row vector $c$. Moreover, $b$ is an $m$-dimensional
column vector, and $A$ is an $m \times n$ matrix.
The feasible region of $\mathcal{G}$ is given by $\Omega = \{x:\, Ax \geq b, \, x \in \{0,1\}^n\}$.
We highlight that, although the results of this work are presented by the assumption of $\mathcal{G}$ being a minimization problem,
they also hold for interval 0-1 min-max regret versions of maximization problems, with minor modifications.

Now, let $\mathcal{R}$ be an interval min-max regret RO version of $\mathcal{G}$, where a continuous cost interval $[l_i,u_i]$, with $l_i,u_i \in \mathbb{Z}_+$ and $l_i \leq u_i$, is associated with each binary variable $x_i$, $i =1,\ldots, n$. The following definitions describe $\mathcal{R}$ formally.

\begin{definition}
A \emph{scenario} $s$ is an assignment of costs to the binary variables, i.e., a cost $c_{i}^s \in [{l}_{i},{u}_{i}]$ is fixed for all $\, x_i$, $i =1,\ldots, n$.
\end{definition}

Let $\mathcal{S}$ be the set of all possible cost scenarios, which consists of the cartesian product of the continuous
intervals $[l_i,u_i]$, $i = 1, \ldots, n$.
The cost of a solution $x \in \Omega$ in a scenario $s \in \mathcal{S}$ is given by $c^sx = \sum\limits_{i = 1}^{n}{c^s_ix_i}$.

\begin{definition}
A solution $opt(s) \in \Omega$ is said to be \emph{optimal} for a scenario $s \in \mathcal{S}$ if it has the
smallest cost in $s$ among all the solutions in $\Omega$, i.e., $opt(s) = \arg\min\limits_{x \in \Omega}{c^sx}$.
\end{definition}

\begin{definition}
The \emph{regret (robust deviation)} of a solution $x \in \Omega$ in a scenario $s \in \mathcal{S}$, denoted by $r_x^s$,
is the difference between the cost of $x$ in $s$ and the cost of $opt(s)$ in $s$, i.e.,
$r_x^s = c^sx - c^s{opt(s)}$.
\end{definition}

\begin{definition}
The \emph{robustness cost} of a solution $x \in \Omega$, denoted by $R_x$, is the maximum
regret of $x$ among all possible scenarios, i.e., $R_x = \max\limits_{s \in \mathcal{S}}{r^{s}_x}$.
\end{definition}

\begin{definition}
A solution $x^* \in \Omega$ is said to be \emph{robust} if it has the smallest robustness cost
among all the solutions in $\Omega$, i.e., $x^* = \arg\min\limits_{x \in \Omega}{R_{x}}$.
\end{definition}

\begin{definition}
The \emph{interval min-max regret problem} $\mathcal{R}$ consists in finding a robust solution $x^* \in \Omega$.
\end{definition}

For each scenario $s \in \mathcal{S}$, let $\mathcal{G}(s)$ denote the corresponding problem $\mathcal{G}$ under cost vector $c^s \in \mathbb{R}_+^n$,
\emph{i.e.}, the problem of finding an optimal solution $opt(s)$ for $s$.
Also consider $y$, an $n$-dimensional vector of binary variables.
Then, $\mathcal{R}$ can be generically modeled as follows.

\begin{eqnarray}
  \mbox{$(\mathcal{R})\quad$}\min && \max\limits_{s\in \mathcal{S}}\;(c^sx - \overbrace{\min\limits_{y \in \Omega}{c^sy}}^{(\mathcal{G}(s))}) \label{pr00}\\
  s.t. && x \in \Omega.
\end{eqnarray}

The basic result presented below has been explicitly proved for several interval min-max regret problems (see, \emph{e.g.}, \cite{Karasan01,MontemanniBarta07,YaKaPi01})
and generalized for the case of interval 0-1 min-max regret problems \cite{Aissi09} (also see \cite{Averbakh01}).

\begin{proposition}[Aissi et al. \cite{Aissi09}]
\label{prop01}
The regret of any feasible solution ${x} \in \Omega$ is maximum in the scenario $s({{x}})$ induced by ${x}$, defined as follows:

\[
\textrm{for all } \; i \in \{1,\ldots,n\},\quad c^{s({x})}_i = \left\{
\begin{array}{lll}
        \; u_i, & \textrm{if ${x}_i = 1$,} \\
        \; l_i, & \textrm{if ${x}_i = 0$.} 
\end{array}\right. \]

\end{proposition}

From Proposition~\ref{prop01}, $\mathcal{R}$ can be rewritten as
\begin{eqnarray}
  \mbox{$(\tilde{\mathcal{R}})\quad$}\min && \Big(c^{s(x)}x - \min\limits_{y \in \Omega}{c^{s(x)}y}\Big) \label{r00} \\
  s.t. && x \in \Omega \label{r01}.
\end{eqnarray}

One may note that the inner minimization in (\ref{r00}) does not define an LP problem, but a BILP one. Since, in general, there is no guarantee of
integrality in solving the linear relaxation of this problem, we cannot represent it by means of extreme points and extreme rays, as in a classical Benders'
reformulation \cite{Benders62}.
Alternatively, we reformulate $\tilde{\mathcal{R}}$ by adding a free variable $\rho$ and linear constraints that explicitly bound $\rho$ with
respect to all the feasible solutions that $y$ can represent.
The resulting MILP formulation (see, \emph{e.g.}, \cite{Aissi09}) is provided from \eqref{f00} to \eqref{f03}.
\begin{eqnarray}
  \mbox{$(\mathcal{F})\quad$}\min && (\sum\limits_{i = 1}^{n}{u_ix_i} - \rho) \label{f00} \\
  s.t. && \rho \leq {\sum\limits_{i = 1}^{n}{(l_i + (u_i - l_i)x_i)\bar{y}_i}} \quad \forall\, \bar{y} \in \Omega, \label{f01} \\
        && x \in \Omega,   \label{f02} \\
        && \rho \mbox{ free}. \label{f03}
\end{eqnarray}

Constraints (\ref{f01}) ensure that $\rho$ does not exceed the value related to the inner minimization in (\ref{r00}).
Note that, in (\ref{f01}), $\bar{y}$ is a constant vector, one for each solution in $\Omega$. These constraints are tight whenever
$\bar{y}$ is optimal for the classical counterpart problem $\mathcal{G}$ in the scenario $s(x)$.
Constraints (\ref{f02}) define the feasible region referred to the $x$ variables, and constraint (\ref{f03}) gives the domain of the variable $\rho$.
Notice that the feasibility of $\mathcal{F}$ solely relies on the feasibility of the corresponding classical optimization problem $\mathcal{G}$.
Thus, for simplicity, we assume that $\mathcal{F}$ is feasible in the remainder of this work.

The number of constraints (\ref{f01}) corresponds to the number of feasible solutions in $\Omega$.
As the size of this region may grow exponentially with the number of binary variables,
this fomulation is particularly suitable to be handled by decomposition methods, such as the logic-based Benders' decomposition framework detailed below.

\subsection{Logic-Based Benders' Algorithm}
\label{s_benders}

The logic-based Benders' algorithm here described relies on the fact that, since several of constraints~(\ref{f01}) might be inactive at optimality,
they can be generated on demand whenever they are violated.
In this sense, given a set $\Gamma \subseteq \Omega$, $\Gamma \neq \emptyset$, consider the \emph{relaxed robustness cost} metric defined as follows. 

\begin{definition}
A solution $opt(s,\Gamma) \in \Gamma$ is said to be $\Gamma$-\emph{relaxed optimal} for a scenario $s \in \mathcal{S}$ if it has the
smallest cost in $s$ among all the solutions in $\Gamma$, i.e., $opt(s,\Gamma) = \arg\min\limits_{x \in \Gamma}{c^{s}x}$.
\end{definition}

\begin{definition}
The $\Gamma$-\emph{relaxed robustness cost} of a solution $x \in \Omega$, denoted by $R_x^{\Gamma}$, is the difference between the cost
of $x$ in the scenario $s(x)$ induced by $x$ and the cost of a $\Gamma$-relaxed optimal solution $opt(s(x),\Gamma)$ in $s(x)$,
i.e., $R_x^{\Gamma} = c^{s(x)}x - c^{s(x)}opt(s(x),\Gamma)$.
\end{definition}

\begin{proposition}
\label{prop_lb}
For any $\Gamma \subseteq \Omega$, $\Gamma \neq \emptyset$, and any solution $x \in \Omega$, the $\Gamma$-relaxed robustness cost 
$R^{\Gamma}_x$ of $x$ gives a lower bound on the robustness cost $R_x$ of $x$.
\end{proposition}

\begin{proof}
Consider a set $\Gamma \subseteq \Omega$, $\Gamma \neq \emptyset$, and a solution $x \in \Omega$.
According to Proposition~\ref{prop01}, the robustness cost of $x$ is given by $R_x = r_x^{s(x)} = c^{s(x)}x -
c^{s(x)}opt(s(x))$, where $opt(s(x))$ is an optimal solution for the scenario $s(x)$ induced by $x$. By definition, the
$\Gamma$-relaxed robustness cost of $x$ is given by $R_x^{\Gamma} = c^{s(x)}x - c^{s(x)}opt(s(x),\Gamma)$, where $opt(s(x),\Gamma)$
is a $\Gamma$-relaxed optimal solution for $s(x)$.
Notice that $c^{s(x)}{opt(s(x))} \leq c^{s(x)}x'$ for all $x' \in \Omega$, including $opt(s(x),\Gamma)$. Therefore,
\begin{eqnarray}
 R_x^{\Gamma} = c^{s(x)}x - c^{s(x)}{opt(s(x),\Gamma)} \leq c^{s(x)}x - c^{s(x)}{opt(s(x))} = R_x.
\end{eqnarray}\qed
\end{proof}

\begin{proposition}
\label{prop_conv1}
If $\Gamma = \Omega$, then, for any solution $x \in \Omega$, it holds that $R^{\Gamma}_x = R_x$.
\end{proposition}

\begin{proof}
Consider the set $\Gamma = \Omega$ and a solution $x \in \Omega$. In this case, a $\Gamma$-relaxed optimal solution
$opt(s(x),\Gamma)$ for $s(x)$ is also an optimal solution $opt(s(x))$ for this scenario.
Therefore, considering Proposition~\ref{prop01},
\begin{eqnarray}
 R_x^{\Gamma} = c^{s(x)}x - c^{s(x)}{opt(s(x),\Gamma)} = c^{s(x)}x - c^{s(x)}{opt(s(x))} = R_x.
\end{eqnarray}\qed
\end{proof}

\begin{definition}
A solution $\tilde{x}^* \in \Omega$ is said to be $\Gamma$-\emph{relaxed robust} if it has the smallest
$\Gamma$-relaxed robustness cost among all the solutions in $\Omega$, \emph{i.e.}, $\tilde{x}^* = \arg\min\limits_{x \in \Omega}{R_{x}^{\Gamma}}$.
\end{definition}

Considering the relaxed metric discussed above, we detail a logic-based Benders' algorithm to solve formulation $\mathcal{F}$, given by (\ref{f00})-(\ref{f03}).
The procedure is described in Algorithm~\ref{Benders01}.
Let ${\Omega}^{\psi} \subseteq {\Omega}$ be the set of solutions $\bar{y} \in \Omega$ (Benders' cuts) available at an iteration $\psi$. Also let
$\mathcal{F}^{\psi}$ be a relaxed version of $\mathcal{F}$ in which constraints (\ref{f01}) are replaced by
\begin{equation}
 \rho \leq \sum\limits_{i = 1}^{n}{( l_{i} + (u_{i} - l_{i})x_{i}){\bar y}_{i}} \quad \forall\, {\bar y} \in {\Omega}^{\psi}. \label{f04}
\end{equation}

Thus, the relaxed problem $\mathcal{F}^{\psi}$, called \emph{master problem}, is defined by (\ref{f00}), (\ref{f02}), (\ref{f03}) and (\ref{f04}).
One may observe that $\mathcal{F}^{\psi}$ is precisely the problem of finding a $\Gamma$-relaxed robust solution, with $\Gamma = \Omega^{\psi}$. 

Let $ub^{\psi}$ keep the best upper bound found (until an iteration $\psi$) on the solution of $\mathcal{F}$. 
Notice that, at the beginning of Algorithm~\ref{Benders01}, ${\Omega}^{1}$ contains the initial Benders' cuts available, whereas $ub^1$ keeps
the initial upper bound on the solution of $\mathcal{F}$.
In this case, ${\Omega}^{1}=\emptyset$ and $ub^1 := + \infty$.
At each iteration ${\psi}$, the algorithm obtains a solution by solving
a corresponding master problem  $\mathcal{F}^{\psi}$ and seeks a constraint (\ref{f01}) that is most violated by this solution.
Initially, no constraint (\ref{f04}) is considered, since ${\Omega}^{1}=\emptyset$. An initialization step is then necessary
to add at least one solution to ${\Omega}^{1}$, thus avoiding unbounded solutions during the first resolution of the master problem.
To this end, it is computed an optimal solution for the worst-case scenario $s_u$, in which $c^{s_u} = u$ (Step \rom{1}, Algorithm~\ref{Benders01}).

\begin{algorithm}[!ht]
\caption{Logic-based Benders' algorithm.}
\KwIn{Cost intervals $[l_i,u_i]$ referred to $x_i$, $i =1,\ldots, n$.}
\KwOut{$({\bar x}^*,R^*)$, where ${\bar x}^*$ is a robust solution for $\mathcal{F}$, and $R^*$ is its corresponding robustness cost.}
  $\psi := 1$; $ub^1 := + \infty$; ${\Omega}^{1} := \emptyset$;\\
 \textbf{Step \rom{1}. (Initialization)} \\
 Find an optimal solution ${\bar y}^1 = opt(s_{u})$ for the worst-case scenario $s_{u}$;\\
  ${\Omega}^{1} := {\Omega}^{1} \cup \{{\bar y}^1\}$;\\
  \textbf{Step \rom{2}. (Master problem)}\\ Solve the relaxed problem $\mathcal{F}^{\psi}$, obtaining a solution $({\bar x}^{\psi},{\bar \rho}^{\psi})$;\\
  \textbf{Step \rom{3}. (Separation subproblem)}\\ Find an optimal solution ${\bar y}^{\psi} = opt(s({{\bar x}^{\psi}}))$ for the scenario
  $s({{\bar x}^{\psi}})$ induced by ${\bar x}^{\psi}$
  and use it to compute $R_{{\bar x}^{\psi}}$, the robustness cost of ${\bar x}^{\psi}$;\\
 \textbf{Step \rom{4}. (Stopping condition)} \\
 $lb^{\psi} := {u}{\bar x}^{\psi} - {\bar \rho}^{\psi}$;\\
 \If{$lb^{\psi} \geq {R_{{\bar x}^{\psi}}}$}
 {
  ${\bar x}^* := {\bar x}^{\psi}$;\\
  $R^* := {R_{{\bar x}^{\psi}}}$;\\
  Return $({\bar x}^{*},R^*)$;
  }
 \Else{
 $ub^{\psi} := \min \{ub^{\psi}, R_{{\bar x}^{\psi}}\}$;\\
 $ub^{\psi+1} := ub^{\psi}$;\\
 ${\Omega}^{\psi+1} := {\Omega}^{\psi} \cup \{{\bar y}^{\psi}\}$;\\
 $\psi := \psi + 1$;\\
 Go to Step \rom{2};}

    \label{Benders01}
\end{algorithm}

After the initialization step, the iterative procedure takes place. At each iteration ${\psi}$, the corresponding relaxed problem $\mathcal{F}^{\psi}$
is solved (Step \rom{2}, Algorithm \ref{Benders01}), obtaining a solution $({\bar x}^{\psi},{\bar \rho}^{\psi})$.
Then, the algorithm checks if $({\bar x}^{\psi},{\bar \rho}^{\psi})$ violates any constraint (\ref{f01}) of the original problem $\mathcal{F}$, \emph{i.e.},
if there is a constraint (\ref{f04}) that should have been considered in $\mathcal{F}^{\psi}$ and was not. For this purpose, it is solved a \emph{separation
subproblem} that computes $R_{{\bar x}^{\psi}}$ (the actual robustness cost of ${\bar x}^{\psi}$) by finding an optimal solution
${\bar y}^{\psi} = opt(s({{\bar x}^{\psi})})$ for the scenario $s({{\bar x}^{\psi})}$ induced by ${\bar x}^{\psi}$
(see Step \rom{3}, Algorithm \ref{Benders01}). Notice that the separation subproblems involve solving a classical optimization problem $\mathcal{G}({{\bar x}^{\psi}})$,
\emph{i.e.}, problem $\mathcal{G}$, given by (\ref{bilp00})-(\ref{bilp02}), in the scenario $s({{\bar x}^{\psi}})$.

Let $lb^{\psi} = {u}{{\bar x}^{\psi}} - {\bar \rho}^{\psi}$ be the value of the objective function in (\ref{f00}) related
to the solution $({\bar x}^{\psi},{\bar \rho}^{\psi})$ of the current master problem $\mathcal{F}^{\psi}$. Notice that, considering
$\Gamma = \Omega^{\psi}$, $lb^{\psi}$ corresponds to the $\Gamma$-relaxed robustness cost of ${\bar x}^{\psi}$. Thus, according to
Proposition~\ref{prop_lb}, $lb^{\psi}$ gives a lower (dual) bound on the solution of $\mathcal{F}$. Moreover, since ${\bar x}^{\psi}$ is a
feasible solution in $\Omega$, its robustness cost $R_{{\bar x}^{\psi}}$ gives an upper (primal) bound on the solution of $\mathcal{F}$.
Accordingly, if $lb^{\psi}$ reaches $R_{{\bar x}^{\psi}}$, the algorithm stops. Otherwise,
$ub^{\psi}$ and $ub^{\psi+1}$ are both set to the best upper bound found by the algorithm until the iteration $\psi$. In addition, a new constraint
(\ref{f04}) is generated from ${\bar y}^{\psi}$ and added to
$\mathcal{F}^{\psi+1}$ by setting ${\Omega}^{\psi+1} := {\Omega}^{\psi} \cup \{{\bar y}^{\psi}\}$ (see Step \rom{4} of
Algorithm \ref{Benders01}). Notice that the algorithm stops when the value ${\bar \rho}^{\psi}$ corresponds to the cost of
${\bar y}^{\psi} = opt(s({{\bar x}^{\psi})})$ in the scenario $s({{\bar x}^{\psi}})$, \emph{i.e.}, the optimal solution for $\mathcal{F}^{\psi}$ is also feasible
(and, therefore, optimal) for the original problem $\mathcal{F}$.
The convergence of the algorithm is ensured by Proposition~\ref{prop_conv1} and the following results.

\begin{lemma}
 \label{lemma_conv2}
 Every separation subproblem that arises during the execution of Algorithm~\ref{Benders01} is feasible.
\end{lemma}

\begin{proof}
Assuming $\mathcal{F}$ feasible, we must have $\Omega \neq \emptyset$. This implies the existence of
at least one feasible solution for every scenario $s \in S$, and, thus, any classical problem $\mathcal{G}$ that arises while executing Algorithm~\ref{Benders01} is
feasible.\qed
\end{proof}

\begin{proposition}
 \label{prop_conv3}
 At each iteration $\psi \geq 1$ of Algorithm~\ref{Benders01}, if the stopping condition is not satisfied, then
 the resolution of the corresponding separation subproblem leads to a new solution
 $\bar{y}^{\psi} \in \Omega \backslash \Omega^{\psi}$.
\end{proposition}

\begin{proof}
 Consider an iteration $\psi \geq 1$ of Algorithm~\ref{Benders01} and assume, by contradiction, that (\rom{1})
 the stopping condition is not satisfied, and (\rom{2}) the resolution of the separation subproblem of the current iteration $\psi$ does not
 lead to a solution in $\Omega \backslash \Omega^{\psi}$.
 Let $({\bar x}^{\psi},{\bar \rho}^{\psi})$ be the solution obtained from the resolution of the corresponding master problem $\mathcal{F}^{\psi}$.
From Lemma~\ref{lemma_conv2}, the subproblem referred to iteration $\psi$ is feasible, and, thus, its resolution leads to an optimal solution
${\bar y}^{\psi} = opt(s({{\bar x}^{\psi})}) \in \Omega$ for the scenario $s({{\bar x}^{\psi}})$ induced by ${{\bar{x}}^{\psi}}$.
From assumption (\rom{1}), we must have $lb^{\psi} < R_{{\bar x}^{\psi}}$. Considering Proposition~\ref{prop01} and letting
$\Gamma = \Omega^{\psi}$, we have that $lb^{\psi} = R_{{\bar x}^{\psi}}^{\Gamma}$, and, moreover,
\begin{align}
 c^{s({{\bar x}^{\psi}})}{{\bar x}^{\psi}} - c^{s({{\bar x}^{\psi}})}{opt(s({{\bar x}^{\psi}}),\Gamma)} & = R_{{\bar x}^{\psi}}^{\Gamma} \label{prop_conv3.0}\\
  & = lb^{\psi} \\
  & < R_{{\bar x}^{\psi}}\\ 
  & = c^{s({{\bar x}^{\psi}})}{{\bar x}^{\psi}} - c^{s({{\bar x}^{\psi}})}{opt(s({{\bar x}^{\psi}}))} \label{prop_conv3.1},
\end{align}

\noindent where ${opt(s({{\bar x}^{\psi}}),\Gamma)}$ is a $\Gamma$-relaxed optimal solution for $s({{\bar x}^{\psi}})$, and ${opt(s({{\bar x}^{\psi}}))}$
is an optimal solution for $s({{\bar x}^{\psi}})$. From (\ref{prop_conv3.0})-(\ref{prop_conv3.1}), we obtain
\begin{equation}
\label{prop_conv3.2}
c^{s({{\bar x}^{\psi}})}{{opt}(s({{\bar x}^{\psi}}),\Gamma)} > c^{s({{\bar x}^{\psi}})}{opt(s({{\bar x}^{\psi}}))}.
\end{equation}

Notice that, since $\bar{y}^{\psi}$ is also an optimal solution for $s({{\bar x}^{\psi}})$, it follows, from (\ref{prop_conv3.2}), that 
 \begin{equation}
 \label{prop_conv3.3}
 c^{s({{\bar x}^{\psi}})}{opt(s({{\bar x}^{\psi}}),\Gamma)} > c^{s({{\bar x}^{\psi}})}{opt(s({{\bar x}^{\psi}}))} = c^{s({{\bar x}^{\psi}})}{{\bar y}^{\psi}}.
 \end{equation}

Nevertheless, as ${opt(s({{\bar x}^{\psi}}),\Gamma)}$ is a $\Gamma$-relaxed optimal solution for $s({{\bar x}^{\psi}})$, and,
from Lemma~\ref{lemma_conv2} and assumption (\rom{2}), ${\bar y}^{\psi}$ belongs to $\Omega^{\psi} = \Gamma$, we also have that 
 \begin{equation}
 \label{prop_conv3.4}
  c^{s({{\bar x}^{\psi}})}{opt(s({{\bar x}^{\psi}}),\Gamma)} \leq c^{s({{\bar x}^{\psi}})}{{\bar y}^{\psi}},
 \end{equation}

\noindent which, considering (\ref{prop_conv3.3}), defines a contradiction.\qed
\end{proof}

\begin{theorem}
Algorithm~\ref{Benders01} solves problem $\mathcal{F}$ at optimality within a finite number of iterations.
\end{theorem}

\begin{proof}
As $\Omega$ is defined in terms of binary variables, it consists of a finite discrete set of solutions. Thus, the convergence of
Algorithm~\ref{Benders01} is guaranteed by Proposition~\ref{prop_conv1} and Proposition~\ref{prop_conv3}. \qed
\end{proof}

\section{Concluding Remarks}
In this work, we presented the first formal proof of the finite convergence of state-of-the-art logic-based Benders' decomposition algorithms for
a class of robust optimization problems, namely interval 0-1 min-max regret problems.
These algorithms were generically described by means of a logic-based Benders' decomposition framework, which was proved to converge to
an optimal solution in a finite number of iterations.

\bibliographystyle{splncs04}
 \bibliography{main}
 
\end{document}